\documentclass[a4paper,UKenglish,cleveref, autoref, thm-restate]{lipics-v2021}

\usepackage[utf8]{inputenc}
\usepackage[T1]{fontenc}

\usepackage{tabularx}%
\usepackage{underscore}%
\usepackage{booktabs}%
\usepackage{amsfonts}%
\usepackage{amsmath}%
\usepackage{multirow}%
\usepackage{listings}%
\usepackage{lmodern}%
\usepackage{bookmark}%
\usepackage{color}%
\usepackage{sidecap}%
\usepackage[gen]{eurosym}%
\usepackage{comment}%
\usepackage{haskell}

\usepackage{tikz}%
\usetikzlibrary{arrows.meta}%
\usetikzlibrary{shapes}%
\usetikzlibrary{shapes.misc}%
\usetikzlibrary{backgrounds}%
\usetikzlibrary{calc}%

\usepackage{graphicx}
\newcommand\ind{\mbox{}\quad}

\newenvironment{isabelle} {%
  \begin{list}{}{\leftmargin2ex}
  \item\relax\(}{\)
  \end{list}
}%
\newcommand\isa[1]{\ensuremath{\texttt{#1}}}%
\newcommand\isakwd[1]{\textbf{#1}}%
\newcommand\isaIkwd[1]{\textbf{\textit{#1}}}%

\def\R{\mathbb{R}}%
\def\F{\mathbb{F}}%
\def\N{\mathbb{N}}%
\def\Q{\mathbb{Q}}%
\def\B{\mathbb{B}}%

\DeclareMathOperator*{\maximize}{maximise}
\DeclareMathOperator*{\minimize}{minimise}
\DeclareMathOperator*{\subjectto}{subject\ to:\ }

\newcommand{\defines}{:=}

\newcommand{\rllpoly}{\texttt{list\_to\_lpoly}}
\newcommand{\rvlpoly}{\texttt{vec\_to\_lpoly}}
\newcommand{\vflpoly}{\texttt{lpoly\_to\_vec}}
\newcommand{\myvec}{\texttt{vec}}
\newcommand{\dimpoly}{\texttt{dim\_poly}}
\newcommand{\coeff}{\texttt{coeff}}
\newcommand{\mtlpol}{\texttt{matrix\_to\_lpolies}}
\newcommand{\poltmt}{\texttt{lpolies\_to\_matrix}}
\newcommand{\twomatninter}{\texttt{two\_block\_non\_interfere}}
\newcommand{\xcgeqyb}{\texttt{xc\_geq\_yb}}
\newcommand{\figeq}{\texttt{from\_ind\_geq}}
\newcommand{\satprimal}{\texttt{sat\_primal}}
\newcommand{\satdual}{\texttt{sat\_dual}}
\newcommand{\optLP}{\texttt{optimal\_LP}}
\newcommand{\vappend}{\ensuremath{\mathrel{\texttt{@}_v}}}
\newcommand{\listAppend}{\ensuremath{\mathrel{\texttt{@}}}}
\newcommand{\zerovn}[1]{\ensuremath{\texttt{0}_v^{#1}}}
\newcommand{\dimvec}[1]{\ensuremath{\texttt{dim}\ #1}}
\newcommand{\dimrow}[1]{\ensuremath{\texttt{row}\ #1}}
\newcommand{\dimcol}[1]{\ensuremath{\texttt{col}\ #1}}
\newcommand{\typeType}[1]{\ensuremath{\texttt{#1}}}
\newcommand{\dotP}{\ensuremath{\bullet}}

\newcommand{\matVecMul}{\ensuremath{\cdot_v}}
\newcommand{\vecMatMul}{\ensuremath{\cdot^v}}
\newcommand{\abstractLP}{\texttt{abstract\_LP}}
\newcommand{\createOptSol}{\texttt{create\_optimal\_solution}}
\newcommand{\matgeqeq}{\texttt{mat\_leq\_eqc}}
\newcommand{\optimize}{\texttt{maximize}}
\newcommand{\transpose}[1]{\ensuremath{#1^{T}}}
\newcommand{\MAPsol}{\ensuremath{\texttt{SOL}}}
\newcommand{\MAPsolNTH}[1]{\ensuremath{\MAPsol_{[#1]}}}
\newcommand{\vecNTH}[2]{\ensuremath{#1\ \$\ #2}}
\newcommand{\singletonList}[1]{\ensuremath{[#1]}}
\newcommand{\simplex}{\texttt{simplex}}
\newcommand{\splitmn}{\texttt{split\_nm}}
\newcommand{\minLP}{\ensuremath{\texttt{min\_LP}}}
\newcommand{\maxLP}{\ensuremath{\texttt{max\_LP}}}
\newcommand{\unsat}{\ensuremath{\texttt{Unsat}}}
\newcommand{\sat}{\ensuremath{\texttt{Sat}}}

\newcommand{\listNTH}[2]{\ensuremath{#1!#2}}
\newcommand{\listFromToExcl}[2]{\ensuremath{[#1..{\setlength{\thickmuskip}{0mu}<}#2]}}

\newcommand{\listComprehension}[3]{\ensuremath{[\ #1\ .\ \ #2
    \leftarrow #3\ ]}}

\newcommand{\PwEq}{\ensuremath{\mathrel{=_{pw}}}}
\newcommand{\PwGeq}{\ensuremath{\mathrel{\geq_{pw}}}}
\newcommand{\PwLeq}{\ensuremath{\mathrel{\leq_{pw}}}}
\newcommand{\ConstructEq}{\ensuremath{\mathrel{[=]}}}
\newcommand{\ConstructLeq}{\ensuremath{\mathrel{[\leq]}}}
\newcommand{\ConstructGeq}{\ensuremath{\mathrel{[\geq]}}}

\newcommand{\sumType}{\ensuremath{\mathrel{\texttt{+}}}}

\newcommand{\some}{\ensuremath{\texttt{Some}}}
\newcommand{\none}{\ensuremath{\texttt{None}}}

\def\nats{\mathbb{N}}%

\usepackage{prettyref}%
\newrefformat{sec}{Section~\ref{#1}}%
\newrefformat{ex}{Example~\ref{#1}}%
\newrefformat{syscons}{System of Constraints~\ref{#1}}%
\newrefformat{fig}{Figure~\ref{#1}}%
\newrefformat{tab}{Table~\ref{#1}}%
\newrefformat{eq}{Equation~\ref{#1}}%
\newrefformat{constraint}{Constraint~\ref{#1}}%
\newrefformat{lp}{Linear Program~\ref{#1}}%
\newrefformat{subsec}{Subsection~\ref{#1}}%
\newrefformat{locale}{Locale~\ref{#1}}%
\newrefformat{def}{Definition~\ref{#1}}%
\newrefformat{thm}{Theorem~\ref{#1}}%
\newrefformat{lem}{Lemma~\ref{#1}}%
\newrefformat{subsec}{Subsection~\ref{#1}}%

\lstdefinestyle{bash}{%
  basicstyle=\ttfamily\normalsize,%
  commentstyle=\ttfamily\normalsize,%
  numbers=none,%
  xleftmargin=15pt,%
  xrightmargin=15pt,%
  nolol=true,%
  language=bash,%
  showstringspaces=false%
}

\lstnewenvironment{bash}[1][fullflexible]{\lstset{style=bash,columns=#1}}{}

\title{Linear Programming in Isabelle/HOL} 

\titlerunning{Linear Programming in Isabelle/HOL} 

\author{Julian Parsert}{University of Oxford, United Kingdom \and University of Innsbruck, Austria \and \url{http://www.parsert.com} }{julian.parsert@gmail.com}{https://orcid.org/0000-0002-5113-0767}{}

\authorrunning{J. Parsert} 

\Copyright{Julian Parsert} 

\ccsdesc[500]{Theory of computation~Logic and verification}
\keywords{Linear Programming, Optimisation, Interactive
    Theorem Proving, Isabelle/HOL.} 

\acknowledgements{We thank René
    Thiemann and Cezary Kaliszyk for their help. This work is
    supported by the European Research Council (ERC) grant no
    714034 \textit{SMART}.}

\nolinenumbers 

\hideLIPIcs  

\begin{document}

\maketitle

\begin{abstract}
  Linear programming describes the problem of optimising a linear
  objective function over a set of constraints on its variables. In
  this paper we present a solver for linear programs implemented in
  the proof assistant Isabelle/HOL. This allows formally proving its
  soundness, termination, and other properties. We base these results
  on a previous formalisation of the simplex algorithm which does not
  take optimisation problems into account. Using the weak duality
  theorem of linear programming we obtain an algorithm for solving
  linear programs. Using Isabelle's code generation mechanism we can
  generate an external solver for linear programs.
\end{abstract}

\section{Introduction}
Linear programming is a methodology for solving certain types of
optimisation problems. Linear programming as a part of operations
research also has applications in many areas outside of pure
mathematics and computer science. Examples of applications include but
are not limited to: finance, transportation, management, etc. In
computer science linear programming can be used in, for example,
network optimisation and integer transition systems. Finally, our
motivation for this work is its use in game theory. One can express a
two player zero-sum game using a linear program. Solving this linear
program is equivalent to solving the original game. Hence, linear
programming can be used to solve two-player zero sum
games~\cite{maschler:solan:zamir:2013,winston2004operations}.

Due to its large amount of use cases many software suites ship with a
solver for linear programs~\cite{glpk,GleixnerSteffyWolter2012},
including popular software like Microsoft Excel. However, software is
known to have bugs and undesirable behaviour and the aforementioned
tools most certainly are no exception. Therefore, we believe that
formal verification can be a useful tool to increase trust in the
results of linear program solvers, especially when they are applied to
fault critical areas. In this paper we discuss the use of a proof
assistant to formalise the notion of linear programs and an algorithm
for solving them. We use the proof assistant Isabelle/HOL. In
particular, we formalise an algorithm for solving
linear programs based on a previous formalisation of the general
simplex method~\cite{Spasic:FormIncrSimplex}. This algorithm is a reduction
that reduces the optimisation problem to a constraint satisfaction problem.
Our description of this reduction is stated in such a way that Isabelle's code generation
mechanism can be utilised to generate a formally verified Haskell
program which solves linear programs. To summarise, our contributions
are as follows:
\begin{itemize}
\item We formalise linear programs using the proof assistant
  Isabelle/HOL and derive results related to the duality of linear
  optimisation.
\item We describe an algorithm for solving linear programs and prove
  its soundness. This algorithm is stated in a way such that
  Isabelle's code generation mechanism can be used to obtain a
  verified executable program.
\item In order to obtain the aforementioned results, we provide
  translations and equivalences between an Isabelle library for linear
  polynomials and one for linear algebra. We also provide correctness
  results for these translations.
\item Using the generated program we solve an example game as
  a case study to show how games can be solved using linear
  programming.
\end{itemize}
%
\paragraph*{Related Work:} Our work is based on a formalisation of the
general simplex algorithm described
in~\cite{SimplexAFP,Spasic:FormIncrSimplex}. However, the general
simplex algorithm lacks the ability to optimise a function. Boulmé and
Maréchal~\cite{Sylvain:CoqTacForEqualityLinArith} describe a
formalisation and implementation of Coq tactics for linear integer
programming and linear arithmetic over rationals. More closely related
is the formalisation by Allamigeon et
al.~\cite{Allamigeon:FormCvxPolyhedraSimplex} which formalises the
simplex method and related results in Coq. As part of Flyspeck project Obua
and Nipkow~\cite{Obua2009} created a verification mechanism for linear
programs using the HOL computing library and external solvers.
\paragraph*{Outline:}
In \prettyref{subsec:isabelleNotation} we introduce the proof assistant
Isabelle/HOL as well as notation which will be used throughout this
paper. Subsequently, in \prettyref{sec:lp} we give a short overview of linear
programming. In \prettyref{sec:formalisation} we describe the formalisation and
provide more details on the definitions, algorithms and theorems. In
\prettyref{sec:codeGen} we discuss the generated algorithm and some
examples. Finally, in \prettyref{sec:conclusion} we make concluding remarks and
discuss future work.
\subsection{Isabelle/HOL and Notation}\label{subsec:isabelleNotation}
We use the proof assistant Isabelle/HOL, which is based on simply
typed higher-order logic. On top of the simple type system
Isabelle/HOL provides type classes. We will also use Isabelle's
standard \typeType{option}\ type with the constructors \some{} and
\none{} as well as the sum type denoted ``\sumType''. While the
constructors of the sum type are \typeType{Inl} and \typeType{Inr}, we
will use \typeType{Unsat} and \typeType{Sat} instead. We also use
$\N$ and $\Q$ to denote Isabelle's natural and rational numbers as
well as $\alpha\ \typeType{list}$ and
$(\alpha,\ \beta)\ \typeType{mapping}$ for polymorphic lists and
mappings from $\alpha$ to $\beta$. For a mapping $M$ we use the
notation $M_{[i]}$ to denote the value of $i$ in $M$. We use \dotP\ to
denote the dot product between a row vector and a column vector. The
symbols \matVecMul\ and \vecMatMul\ describe the vector-matrix and
matrix-vector multiplication respectively. The symbols $=$, $\leq$,
and $\geq$ retain their standard semantics for scalars, while we use
\PwEq, \PwLeq, and \PwGeq{} to denote the respective pointwise orders
on vectors. Importantly, we also use a \typeType{constraint} type
describing a constraint. The constructors of these are \ConstructEq,
\ConstructLeq, and \ConstructGeq. To see the difference, note that
while $x \leq y$ and $x \PwLeq A$ are of type \typeType{bool}, the
expression $x \ConstructLeq y$ is of type \typeType{constraint} which
is the pair $(x, y)$ in addition with one of the aforementioned
constructors. We use \isa{[m..<n]} and \isa{\{m..<n\}} to denote lists
and sets of elements from $m$ to $n-1$, while dropping the ``\isa{<}''
symbol also includes $n$. Furthermore, $[ f\ x\ .\ i \leftarrow L]$ is
a short notation for \isa{map\ f L}. The function \dimvec{c}, \dimrow{A}, and
\dimcol{A} return the dimension of a vector $c$ and the number of rows
and columns of a matrix $A$. The zero vector of dimension $n$ is
denoted \zerovn{n}. Vector and list concatenation are denoted with the
operators \vappend\ and \listAppend\ while \singletonList{a} is the
singleton list containing $a$. \listNTH{L}{i} and \vecNTH{V}{i} are
list and vector access operators respectively. Both are zero indexed
and are only well defined if $i< \isa{length}\ L$ or $i < \dimvec{V}$.
The code snippets presented in the remainder of the paper
have been formatted for readability omitting brackets, type
annotations, etc.
\section{Linear Programming}\label{sec:lp}
We will give a brief overview of linear programming. Parts of our
formalisation are based on the textbook ``Theory of Linear and Integer
Programming'' by Schrijver~\cite{schrijver1998theory} which we also
recommend for a more detailed presentation of the topic.

A linear program describes the problem where we have an
objective function $f(x_1,\dots,x_n)$ that we want to optimise while the variables
$x_1, \dots x_n$ are subject to a set of constraints. These
constraints can be an equality
\begin{gather}
  \alpha_1x_1 + \dots + \alpha_nx_n = b\label{constraint:equality}
\end{gather}
or a non-strict inequality
\begin{align}
  \alpha_1x_1 + \dots + \alpha_nx_n &\geq b\label{constraint:geq}\\
  \alpha_1x_1 + \dots + \alpha_nx_n &\leq b.\label{constraint:leq}
\end{align}
Note how \prettyref{constraint:equality} is equivalent to the
combination of the Constraints~\ref{constraint:geq}
and~\ref{constraint:leq}. Furthermore, \prettyref{constraint:geq} is
equivalent to $- (\alpha_1x_1 + \dots + \alpha_nx_n) \leq - b$. For
simplicity, we will only consider constraints of
type~\ref{constraint:leq} and~\ref{constraint:equality} the latter of
which we only keep for sake of readability.

Given a set of linear constraints one can pose the question of whether
or not a variable assignment exists that satisfies these constraints.
This decision problem does not take the optimisation of an objective
function into account. An algorithm for deciding this is the
\emph{general simplex} algorithm. In case of success, we can obtain an
arbitrary variable assignment that satisfies all constraints.
\subsection{Linear Optimisation}
After having introduced the general decision problem for the
satisfaction of a list of constraints, we will now introduce linear
programming. In linear programming we are not only interested in
finding an arbitrary satisfying assignment but an assignment which is
optimal with respect to a given (linear) objective function.

More precisely, a linear program is an objective function $f$ which is
subject to a list of constraints $C$:
\[
  f(x_1, \dots, x_n) \defines c_1 * x_1 + \dots + c_n * x_n
\]
\begin{align*}
  C =
  \begin{bmatrix}
    A_{11} * x_1 + \dots + A_{1n} * x_n\leq b_1\\
    A_{21} * x_1 + \dots + A_{2n} * x_n \leq b_2\\
    \vdots \quad \quad \vdots \quad \quad \vdots \\
    A_{m1} * x_1 + \dots + A_{mn} * x_n \leq b_m
  \end{bmatrix}
\end{align*}
This gives rise to a more concise notation for linear programs where
$c_1, \dots, c_n$, $x_1, \dots, x_n$, $b_1, \dots, b_n$ are vectors
and $A_{11}, \dots, A_{mn}$ is a matrix:
\begin{align}\label{lp:primal}
  \begin{split}
    \maximize&\ c \dotP x\\
    \subjectto& A \matVecMul x \PwLeq b
  \end{split}
\end{align}
Solving this linear program is searching for a satisfying assignment
of variables that is optimal with respect to the function $f$.
\emph{Optimal} in this instance means either minimal or maximal.
Hence, we are looking for an assignment for $x_1,\dots, x_n$ such that
it satisfies the constraints and $f(x_1,\dots,x_n)$ is maximal
(minimal) in the set of all $f(y_1,\dots,y_n)$ where $y_1,\dots,y_n$
also satisfy the constraints. A concrete example of a linear program
accompanied by a plot showing the objective function and constraints is
presented in~\prettyref{ex:LP}.
\begin{example}[Linear Program]\label{ex:LP}
  Take the linear program consisting of the following
  objective function
  \begin{align*}
    f(x,y) := 7x + y
  \end{align*}
  and the following set of constraints:
  \begin{equation*}
    2x + y \leq 5 \quad \quad \quad
    -x + 2 y \leq 2 \quad \quad \quad
    \frac{1}{2}x - \frac{1}{2}y \leq \frac{1}{2} \quad \quad \quad
    x + y \geq 1
  \end{equation*}
  Using basic transformations to transform the last inequality to one
  of the form of Inequality~\ref{constraint:leq}, we obtain the
  following matrix $A$, and vectors $b$ and $c$:
  \begin{align}
    A =
    \begin{bmatrix}
      2 & 1\\
      -1 & 2 \\
      \frac{1}{2} & -\frac{1}{2} \\
      -1 & -1 \\
    \end{bmatrix},\ \ b =
    \begin{bmatrix}
      5 \\ 2\\ \frac{1}{2} \\ -1
    \end{bmatrix},\ \ c =
    \begin{bmatrix}
      7 & 1\\
    \end{bmatrix}
    \label{constraints:example}
  \end{align}
  \begin{figure}
    \centering
    \begin{tikzpicture}[scale=1.1]
      \draw[gray!50, thin, step=0.5] (-1,-1.5) grid (7,5);%
      \draw[thick,->] (-1,0) -- (7,0) node[right] {$x$};%
      \draw[thick,->] (0,-1.5) -- (0,5) node[above] {$y$};%

      \foreach \x in {-1,...,5} \draw (\x,0.05) -- (\x,-0.05)
      node[below] {\tiny\x};%
      \foreach \y in {-1,...,4} \draw (-0.05,\y) -- (0.05,\y)
      node[right] {\tiny\y};%

      \fill[blue!50!cyan,opacity=0.3] (0,1) -- (8/5,9/5) -- (2,1) --
      (1,0) -- cycle;%

      \draw (-1,0.5) -- node[above,sloped] {\tiny$-x + 2y \leq 2$}
      (6,4);%
      \draw (-0.5, 1.5) -- node[below,sloped,xshift=1.2cm]
      {\tiny$x + y \geq 1$} (2,-1);%
      \draw (0, 5) -- node[below,sloped,xshift=-1cm]
      {\tiny$2x + y \leq 5$} (3,-1);%
      \draw (0,-1) -- node[below,sloped,xshift=0.5cm]
      {\tiny$0.5x - 0.5 y\leq 0.5$} (5,4);%
      \draw[red] (3/2,4.5) -- node[below,sloped,xshift=-2cm]
      {\tiny$7x + y = 15$} (16/7,-1);%
    \end{tikzpicture}
    \caption{A plot describing the optimisation problem presented in \prettyref{ex:LP}.}
    \label{fig:lp}
  \end{figure}
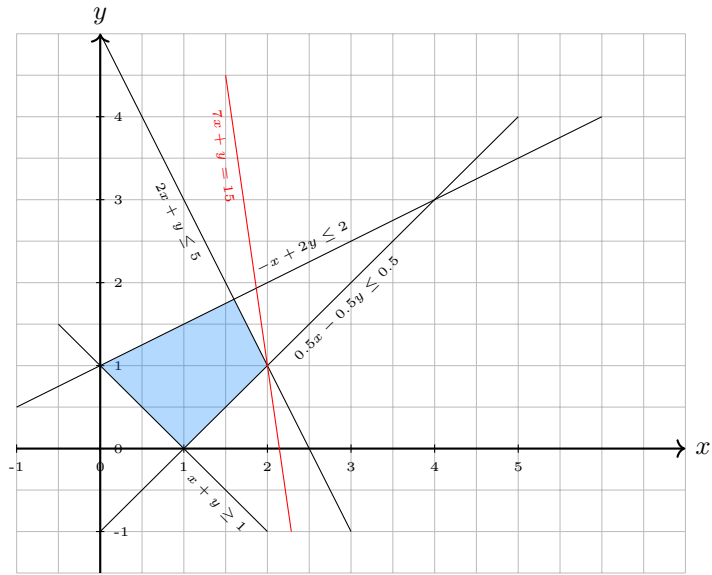
\end{example}
A plot of these inequalities is shown in \prettyref{fig:lp}. The
marked region (blue) is the region satisfying all constraints (i.e.
feasible). The red line (i.e. equation $7x + y = 15$) describes the
objective function going through the point $(2,1)$ which is the optimal
value within the feasible region and therefore the solution to the
linear program described with $A$, $b$, and $c$.
\subsection{Duality}
The duality principle plays an important role in optimisation
problems. It states that any optimisation problem, that is the
\emph{primal} problem, automatically defines a \emph{dual} problem.
Furthermore, a solution to the dual problem is bounded by the solution
of the primal problem. Our algorithm relies on this duality principle.

Since every linear program has a dual, we can write the dual of the
\prettyref{lp:primal}:
\begin{align}\label{lp:dual}
  \begin{split}
    \minimize&\ \transpose{b} \dotP \transpose{y}\\
    \subjectto& y \vecMatMul A \PwEq c,\ 0 \PwLeq y
  \end{split}
\end{align}
Note the switching of $b$ and $c$ in addition to the change from a
maximisation to a minimisation problem. If the original problem is a
maximisation problem, then the value (i.e. value of the optimised function at the
optimal point) of the dual problem is an upper
bound on the value of the primal problem. If the primal problem is a
minimisation, then the dual provides a lower bound. In linear
programming we know that these values are in fact equal. This is known as the Strong Duality Theorem
(\prettyref{thm:duality}).
\begin{theorem}[Strong Duality Theorem]\label{thm:duality}
  Given linear constraints $A$, $b$, and the objective function $c$,
  we obtain $x$ and $y$ as the solutions to the primal and dual linear program, respectively. We can derive the following equality:
  \[
    c \dotP x = y \dotP b
  \]
\end{theorem}
To prove the correctness of our algorithm we only require the weak
duality theorem, which is part of the formalisation and will be
discussed in \prettyref{thm:weakDuality}.
\subsection{Solving Linear Programs}\label{subsec:algorithmThy}
The most common algorithm to solve linear programs is the simplex
algorithm. Unlike the \emph{general simplex} algorithm this alternative simplex
algorithm takes an objective function into account.  However, we use a different
approach. Using the duality theorem we can solve linear programs with the
general simplex algorithm. This can be done by solving the constraints for the
primal program (cf. \prettyref{lp:primal}) and the dual program
(cf. \prettyref{lp:dual}) simultaneously:
\begin{gather}
  A \matVecMul x \PwLeq b,\ y \vecMatMul A \PwEq c,\ y \PwGeq
  0\label{syscons:solution0}
\end{gather}
Due to the dual program constituting an upper bound, we know any
satisfying assignment to $x$ and $y$ must satisfy
$x \dotP c \leq y \dotP b$. Now we can add the constraint
$x \dotP c \geq y \dotP b$ to the Constraints \ref{syscons:solution0}.
Hence, by solving the following constraint satisfaction problem
without explicitly maximising the objective function, we also solve the
linear program:
\begin{gather}
  x \dotP c \geq y \dotP b,\ A \matVecMul x \PwLeq b,\ y \vecMatMul A
  \PwEq c, y \PwGeq 0.\label{syscons:solution}
\end{gather}
Any resulting assignment satisfying the
Constraints~\ref{syscons:solution} also satisfies
$c \dotP x = y \dotP b$. Hence, $x$ solves
the \prettyref{lp:primal}. Furthermore, we also
derive that if a solution to the linear program exists this algorithm
will find it, since no sub-optimal solutions get lost by adding the constraint
$x \dotP c \geq y \dotP b$.
\section{Formalisation}\label{sec:formalisation}
We base our work on two previous formalisation's which are part of the
archive of formal proofs (AFP). The first is due to Spasi{\'{c}} et
al.~\cite{Spasic:FormIncrSimplex,SimplexAFP} who formalise the simplex
algorithm used for checking the satisfiability of linear constraints.
The second one is due to Thiemann et
al.~\cite{DBLP:conf/cpp/Thiemann016,JordanNormalFormAFP} and is a
linear algebra library which allows us to create a relation between
linear polynomials and matrices. All definitions and results in this
section can be found in the formalisation unless specified otherwise.

The formalisation described in~\cite{SimplexAFP} provides a
\emph{sound} and \emph{complete} implementation of the general simplex
algorithm called \simplex{} in Isabelle/HOL. The function \simplex{}
has type
\[
  \typeType{constraint\ list} \Rightarrow \N\ \typeType{list} \sumType
  (\N,\ \Q)\ \typeType{mapping}.
\]
It produces either a variable assignment
$(\N,\ \Q)\ \typeType{mapping}$ (variables are modelled as naturals)
satisfying the constraints or an unsatisfiable core
($\N\ \typeType{list}$). At the current point in time the
\simplex{} algorithm only works on rational numbers. However, due to a
change of the underlying libraries we will also be able to provide
results for real numbers. From now on, we will simply use \simplex{}
as a subroutine without further consideration. For further details on
this formalisation we refer to Marić et al.'s work~\cite{SimplexAFP}.
\subsection{Combining Representations}
When formalising mathematics it is not uncommon to develop theories
that combine existing definitions and representations. In particular,
it is essential to develop methodologies that allow for switching
between representations as some lend themselves better for certain
tasks~\cite{Gonthier:OddOrder}. In our case we use two different
representations for (lists of) linear polynomials.

The first representation is that used by \simplex{}
in~\cite{SimplexAFP}. Here, linear polynomials are defined as
functions mapping variables to their coefficients. As variables are
modelled with natural numbers, polynomials are functions of type
$\N \Rightarrow \Q$ such that for each polynomial $p$ the set $\{x\in \nats .\
p\ x \neq 0 \}$ is finite. The second
representation is one using vectors and matrices. In particular, we
use the linear algebra library described
in~\cite{DBLP:conf/cpp/Thiemann016}. Our motivation for combining
these representations is that the vectors and
matrices make stating and proving some properties easier.

First, we create a mechanism to transform vectors to function type
polynomials. To this end we define a function \rllpoly{} that
translates a list to a polynomial.
\begin{isabelle}
  \isakwd{fun}\ \text{{\normalfont \rllpoly}}\ \isakwd{where}\\
  \ind \text{{\normalfont \rllpoly}}\ cs = \\
  \ind \ind \texttt{sum\_list}\ (\texttt{map2}\ (\lambda i\ c.\
  \texttt{monom}\ c\ i)\ [0..<\texttt{length}\ cs]\ cs)
\end{isabelle}
This function first creates a list of monomials where the index $i$ is
the vector and $\listNTH{cs}{i}$ is the coefficient of variable $i$.
Subsequently, we simply sum this list to obtain the function type
polynomial. Now we get a function that creates linear polynomials from
vectors:
\[
  \rvlpoly\ v = \rllpoly\ (\texttt{list\_of\_vec}\ v)
\]

Going the other direction is a little bit more difficult. First, we
define the dimension of a function-type polynomial $p$ to be $0$ if it
is the zero polynomial and $n$ if $p\ (n-1) \neq 0$ and
$\forall i \geq n.\ p\ i = 0$. Using the vector constructor \isa{vec}
we define a function which transforms a linear polynomial into a
vector;
\[
  \vflpoly\ p = \myvec\ (\dimpoly\ p)\ (\coeff\ p)
\]
The curried function $\isa{coeff}\ p$ is a function that given
$i \in \N$ returns the coefficient of $i$ in the polynomial $p$.

The most important result of combining these representations of
polynomials is \prettyref{thm:inverseRep}.
\begin{theorem}[\rvlpoly{} and \vflpoly{} are (almost)
  inverses]\label{thm:inverseRep}
  For any arbitrary linear polynomial $p$ the equation
  \begin{gather*}
    (\textnormal{\rvlpoly}\ (\textnormal{\vflpoly}\ p)) = p
  \end{gather*}
  holds. Since we lose information about the dimension of the
  original vector $v$ if $v$ ends in a sequence of zeroes we can only show
  the following two results:
  \begin{gather*}
    \vecNTH{\textnormal{\vflpoly}\ (\textnormal{\rvlpoly}\ v)}{i}=
    \vecNTH{v}{i}\\
    \textnormal{\dimvec{(\vflpoly\ (\rvlpoly\ v)))}} \leq
    \textnormal{\dimvec{v}}.
  \end{gather*}
\end{theorem}
Building on these definitions we also define the functions \mtlpol{}
and \poltmt{} which translate a matrix to a list of linear polynomials
and vice versa. Having these two ways of representing linear
polynomials we now use the vector/matrix representation for the remainder of the
paper.
\subsection{Creating Systems of Constraints}
For the algorithm, we need to be able to create and solve a system of
constraints as described in the Constraints as displayed in~\eqref{syscons:solution}.
Since we use the \simplex{} subroutine as a solver, we only need to
worry about creating the system of constraints. The
Constraints in~\eqref{syscons:solution} describe two different vectors $x$
and $y$ with different constraints and a single intersection at the
Constraint $x \dotP c \geq y \dotP b$. Since \simplex{} only allows
for the creation of a single solution of type
$(\N,\ \Q)\ \typeType{mapping}$, we need to synthesise the vectors $x$
and $y$ in certain positions in this mapping. Hence, we introduce
several definitions that allow for the creation of such constraints.

We know that the vector $x$ has to be of length \dimvec{c}\ and the
vector $y$ of length \dimvec{b}. Assuming that the simplex subroutine
terminates successfully with a resulting mapping \MAPsol\ as a SOLution, we create
constraints such that the first \dimvec{c}\ elements in \MAPsol\
constitute the vector $x$ and the elements \MAPsolNTH{\dimvec{c}} to
\MAPsolNTH{\dimvec{c}+\dimvec{b} - 1} the vector $y$.

First, we encode the constraint $y \PwGeq 0$. To keep this as modular
as possible we introduce the function \figeq.
\begin{isabelle}\label{isadef:geq0}
  \isakwd{fun}\ \text{{\normalfont \figeq}}\ :: \N \Rightarrow vector
  \Rightarrow \texttt{constraint list}\ \isakwd{where}\\
  \ind \text{{\normalfont \figeq}}\ ix\ v = [p_{i+ix} \geq v_i.\ i \in
  [0..<\dimvec{v}]]
\end{isabelle}
This allows us to specify a starting index $i$ and a vector $v$, such
that for all $j<\dimvec{v}$, $\MAPsolNTH{i+j} \geq v_i$. Therefore,
given that we synthesise $y$ in the second part of \MAPsol{}\ the
constraint $y \PwGeq 0$ can be expressed as the following:
\begin{gather*}
  \figeq{}\ (\dimvec{c})\ \zerovn{\dimvec{b}}
\end{gather*}

Next, we tackle the two sets of constraints $A \matVecMul x \PwLeq b$
and $y \vecMatMul A \PwEq c$. We will leverage the fact that
$y \vecMatMul A = \transpose{A} \matVecMul \transpose{y}$ in order to
better represent the latter constraint. Since the two constraints are
independent of each other, that is $x$ and $y$ do not interfere, we
first introduce a way of stating them simultaneously. For that we
introduce~\prettyref{def:twoBlockNonInter}. Note that this is a special case of
a block diagonal matrix.
\begin{definition}[Two block non
  interference]\label{def:twoBlockNonInter}
  Given two matrices $A^{m \times n}$ and $B^{a \times b}$ we define a matrix
  \textnormal{\twomatninter},
  \begin{gather*}
    \textnormal{\twomatninter}\ A\ B=
    \begin{bmatrix}
      A_{11}       & \dots & A_{1n}   &   &  &  \\
      \vdots      & \ddots & \vdots  &  & {0} &  \\
      A_{m1}       & \dots & A_{mn}   &  &  &  \\
      &      &         & B_{11} & \dots & B_{1b} \\
      & 0     &         &  \vdots & \ddots & \vdots \\
      &      &         &  B_{a1} & \dots & B_{ab} \\
    \end{bmatrix}
  \end{gather*}
\end{definition}
Using this definition we can show that matrix/vector multiplication of the first
\dimcol{A} elements with $A$ is independent from the multiplication of the last
\dimcol{B} elements with $B$. This notion is captured with
\prettyref{thm:matindep}.
\begin{theorem}\label{thm:matindep}
  Given matrices $A^{m \times n}$ and $B^{a \times b}$ and let $x$ and $y$ be
  $n$ and $b$ dimensional vectors respectively. Then we can show:
  \begin{gather*}
    \textnormal{\twomatninter}\ A\ B \matVecMul (x\vappend y) = (A \matVecMul x)
    \vappend (B \matVecMul y)
  \end{gather*}
\end{theorem}
In order to state the constraints $A \matVecMul x \PwLeq b$ and
$\transpose{A} \vecMatMul \transpose{y} \PwEq c$ simultaneously we create the
following matrix:
\[
  \twomatninter\ A\ \transpose{A}=
  \begin{bmatrix}
    A_{11}       & \dots & A_{1n}   &   &  &  \\
    \vdots      & \ddots & \vdots  &  & {0} &  \\
    A_{m1}       & \dots & A_{mn}   &  &  &  \\
    &      &         & A_{11} & \dots & A_{m1} \\
    & 0     &         &  \vdots & \ddots & \vdots \\
    &      &         &  A_{1n} & \dots & A_{mn} \\
  \end{bmatrix}
\]
Using \mtlpol{} we can translate this matrix to a list of polynomials $L$. Then,
the list of constraints $A \matVecMul x \PwLeq b$ can be generated with:
\begin{gather*}
  \listComprehension{\listNTH{L}{i} \ConstructLeq
    \vecNTH{b}{i}}{i}{\listFromToExcl{0}{\dimvec{b}}}
\end{gather*}
The second list of constraints $\transpose{A} \vecMatMul \transpose{y} \PwEq c$
is:
\begin{gather*}
  \listComprehension{\listNTH{L}{i} \ConstructEq \vecNTH{(b \vappend
      c)}{i}}{i}{\listFromToExcl{\dimvec{b}}{\dimvec{b} + \dimvec{c}}}
\end{gather*}
Combining these definitions we obtain the following Isabelle function which
given a matrix $A$, and vectors $b$, $c$ generates a list of constraints
modelling $A \matVecMul x \PwLeq b$ and
$\transpose{A} \vecMatMul \transpose{y} \PwEq c$.
\begin{isabelle}
  \isakwd{fun}\ \text{{\normalfont \matgeqeq}}\ \isakwd{where}\\
  \ind \text{{\normalfont \matgeqeq}}\ A\ b\ c =\\
  \ind \ind \isaIkwd{let}\ lst = \\
  \ind \ind \ind \mtlpol\ (\twomatninter\ A\ \transpose{A})\\
  \ind \ind \isaIkwd{in}\\
  \ind \ind \listComprehension{\listNTH{lst}{i} \ConstructLeq
    \vecNTH{b}{i}}{i}{\listFromToExcl{0}{\dimvec{b}}} \listAppend\\
  \ind \ind \listComprehension{\listNTH{lst}{i} \ConstructEq \vecNTH{(b \vappend
      c)}{i}}{i}{\listFromToExcl{\dimvec{b}}{\dimvec{b} + \dimvec{c}}}
\end{isabelle}
Due to the use of \twomatninter{} and by \prettyref{thm:matindep} the vectors
$x$ and $y$ are generated in the correct positions.

Finally, we are left with the only constraint where $x$ and $y$ interfere,
$x \dotP c \geq y \dotP b$.
\begin{isabelle}
  \isakwd{fun}\ \text{{\normalfont \xcgeqyb}}\ \isakwd{where}\\
  \ind \text{{\normalfont \xcgeqyb}}\ c\ \ b = \\
  \ind \ind \rvlpoly\ (c\ \vappend \zerovn{\dimvec{b}}) \ConstructGeq \rvlpoly\
  (\zerovn{\dimvec{c}} \vappend b)
\end{isabelle}
This constraint ensures that after extracting $x$ and $y$ from the solution
mapping the following condition holds ($n = \dimvec{c}$ and $m = \dimvec{b}$):
\begin{gather*} [c_0, \dots c_{n-1}, 0_0, \dots,0_{m-1}] \dotP (x \vappend y)
  \geq [0_0, \dots, 0_{n-1}, b_0, \dots, b_{m-1}] \dotP (x \vappend y)
\end{gather*}
This precisely corresponds to the constraint $x \dotP c \geq y \dotP b$. Now we
have defined all the functions necessary to generate the
\prettyref{syscons:solution} by concatenating the lists:
\begin{gather}\label{constraints:system} \singletonList{\xcgeqyb\ c\
    b}\listAppend (\matgeqeq\ A\ b\ c) \listAppend (\figeq{}\ (\dimvec{c})\
  \zerovn{\dimvec{b}})
\end{gather}
Solving this list of constraints with the \simplex\ procedure yields a mapping
\MAPsol\ of $\N$ to $\Q$. Using a simple split function
$\splitmn\ (\dimvec{c})\ (\dimvec{b})\ \MAPsol$ we obtain the pair of vectors
$(x, y)$ which satisfy the Constraints~\ref{constraints:system} and in
turn~\ref{syscons:solution}.
\subsection{Abstract Linear Programming}
Having described a way of expressing the constraints in our setting, we now take
a look at linear programming from an abstract point of view. That is, we will
define necessary definitions and derive results that we use to formally prove
the correctness of our algorithm.

First, we define the abstract notions of satisfying assignments for the primal
and dual problems.
\begin{isabelle}
  \isakwd{definition}\ \text{{\normalfont \satprimal}}\ A\ b =
  \{x.\ A \matVecMul x \PwLeq b \}\\
  \isakwd{definition}\ \text{{\normalfont \satdual}}\ A\ c = \{y.\ y \vecMatMul
  A \PwEq c \land y \PwGeq 0 \}
\end{isabelle}
In addition we define the notion of optimality.
\begin{isabelle}
  \isakwd{definition}\ \text{{\normalfont \optLP}}\ f\ S\ c = \{x \in S.\
  (\forall y \in S.\ f\ (y \dotP c)\ (x \dotP c))\}
\end{isabelle}
Here $f$ is a function of type $\alpha \Rightarrow \alpha \Rightarrow \B$ which
usually defines an order, $S$ is a set of polynomials to optimise over, and $c$
is the objective function represented as a vector. Combining these we get the
maximisation problem \maxLP
\begin{isabelle}
  \optLP\ (\leq)\ (\satprimal\ A\ b)\ c
\end{isabelle}
and its dual minimisation problem \minLP\
\begin{isabelle}
  \optLP\ (\geq)\ (\satdual\ A\ c)\ b.
\end{isabelle}

Next we want to prove the weak duality theorem for \maxLP\ and \minLP.  To this
end we first create an abstract environment using Isabelle's \emph{locale}
mechanism~\cite{Ballarin2014} which also allows us to state assumptions:
\begin{isabelle}\label{locale:lp}
  \isakwd{locale}\ \text{{\normalfont \abstractLP}} =\\
  \ind \isakwd{fixes}\ A\ b\ c\\
  \ind \isakwd{assumes}\ A \in \F^{m \times n}\ \isakwd{and}\ b \in \F^{m}\
  \isakwd{and}\ c \in \F^{n}
\end{isabelle}
Note that since we are only conducting abstract reasoning without a concrete
algorithm yet, we do not restrict ourselves to $\Q$ or $\R$.  Furthermore, the
underlying type-class of $\F$ is a linearly ordered commutative semiring. Within
this environment we can prove \prettyref{lem:weakDualtyAux} and
\prettyref{thm:weakDuality}. The proof of the former can be found in the
formalisation under the name \texttt{weak_duality_aux}.
\begin{lemma}\label{lem:weakDualtyAux}
  If $x$ and $y$ are solutions to the primal and dual problem respectively, then
  \[
    c \bullet x \leq b \bullet y.
  \]
\end{lemma}
\begin{theorem}[Weak Duality Theorem]\label{thm:weakDuality}
  If $x \in \maxLP$ and $y \in \minLP$ then we can show
  \[
    x \dotP c \leq y \dotP b.
  \]
\end{theorem}
\begin{proof}
  From the definition of \maxLP\ we know that $A \matVecMul x \PwLeq b$. And
  since we also know that $y \PwGeq 0$ (from \minLP), we can show
  $y \vecMatMul A \matVecMul x \leq y \dotP b$. Furthermore, from the definition
  of \minLP\ we have $y \vecMatMul A \PwEq c$ and hence
  $c \dotP x = y \vecMatMul A \matVecMul x$. Putting these together we get,
  $c \dotP x = y \vecMatMul A \matVecMul x \leq y \dotP b$. Note that the
  assumptions in the \emph{locale} guarantee that the dimensions of the
  vectors/matrices are correct and allow for the dot product to be commutative,
  thus proving $x \dotP c \leq y \dotP b$.
\end{proof}
\subsection{Final Algorithm}
With the necessary definitions and results we are now able to describe the
algorithm and prove its soundness. As explained above we will be using the
general simplex algorithm \simplex\ which has previously been formalised and
proven correct within Isabelle. With \simplex\ as a subroutine we write the
following function.

\begin{isabelle}
  \isakwd{fun}\ \text{{\normalfont \createOptSol}}\  \isakwd{where}\\
  \ind \createOptSol\ A\ b\ c =\\
  \ind \ind \isaIkwd{let}\ cs = \\
  \ind \ind \ind \singletonList{\xcgeqyb{}\ c\ b}\listAppend \\
  \ind \ind \ind \matgeqeq\ A\ b\ c \listAppend  \\
  \ind \ind \ind \figeq{}\ (\dimvec{c})\ \zerovn{\dimvec{b}}\\
  \ind \ind \isaIkwd{in}\\
  \ind \ind \isaIkwd{case}\ \texttt{simplex}\ cs\ \isaIkwd{of} \ind \\
  \ind \ind \ind |\ \unsat\ S \Rightarrow \unsat\ S \\
  \ind \ind \ind |\ \sat\ S \Rightarrow \sat\ (\splitmn\ (\dimvec c)\ (\dimvec
  b)\ S)
\end{isabelle}
We first create a list of constraints $cs$ which describes the
\prettyref{syscons:solution} (cf. Listing~\ref{constraints:system}).
Subsequently, we use the general simplex algorithm \simplex\ to obtain an
arbitrary variable assignment that satisfies the constraints $cs$.  If a
satisfying assignment exists, we split the resulting assignment and create a
pair of vectors $(x, y)$ where $x$ and $y$ satisfy
the~\prettyref{syscons:solution}. It is important to note that this algorithm
assumes the inputs $A$, $b$, and $c$ to be of the right dimensions (i.e. the
dimensions assumed in the locale \abstractLP).  Furthermore, since this is now a
concrete algorithm the elements of $A$, $b$, and $c$ are assumed to be rational
numbers again.

Finally, in order we get a general algorithm without having to rely on the
locale assumptions, we simply add a check for the dimensions of the input and if
this fails we return \none{} and otherwise the result of \createOptSol.
\begin{isabelle}
  \isakwd{fun}\ \text{{\normalfont \optimize}}\  \isakwd{where}\\
  \ind \optimize\ A\ b\ c =\\
  \ind \ind \isaIkwd{if}\ \dimvec{b} = \dimrow{A} \land \dimvec{c} = \dimrow{A}\ \isaIkwd{then}\\
  \ind \ind \ind \texttt{Some}\ (\createOptSol\ A\ b\ c) \\
  \ind \ind \isaIkwd{else}\ \\
  \ind \ind \ind \texttt{None}
\end{isabelle}
With this algorithm we prove the following theorem without any locale
assumptions, that is without any assumptions outside the one specifically
stated.
\begin{isabelle}
  \isakwd{theorem}\ \text{soundness}\\
  \ind \isakwd{assumes}\ \optimize\ A\ b\ c = \texttt{Some}\ (\texttt{Sat}\ (x, y)) \\
  \ind \isakwd{shows}\ x \in max\_lp\ A\ b\ c
\end{isabelle}
\begin{proof}
  Since the result is $\texttt{Some}\ (\texttt{Sat}\ (x, y))$ we know the
  dimensions of the input were correct. From the constraints created in
  \createOptSol{} we know that for $x$, the following constraints hold:
  \[
    A \matVecMul x \PwLeq b
  \]
  Therefore, by the definition of \satprimal{} we have $x \in \satprimal\ A\
  b$. Conversely, for $y$ we have:
  \[
    y \vecMatMul A \PwEq c\ \text{and}\ y \PwGeq 0
  \]
  Now, by definition of \satdual{} we have $y \in \satdual\ A\ c$.  From
  \prettyref{thm:weakDuality} we derive that no matter if $x$ or $y$ optimise
  their respective objective function, as long as $x \in \satprimal\ A\ b$ and
  $y \in \satdual\ A\ c$ the inequality $x \dotP c \leq y \dotP b$ must
  hold. Finally, due to $\xcgeqyb{}\ c\ b$ we also have that $x$ and $y$ must
  satisfy $x \dotP c \geq y \dotP b$. Combining these we get
  $x \dotP c = y \dotP b$. By \prettyref{lem:weakDualtyAux} all
  $v \in \satprimal\ A\ b$ must obey $v \dotP c \leq y \dotP b$, leading to
  $v \dotP c \leq x \dotP c$. Hence, x is optimal, that is
  $x \in max\_lp\ A\ b\ c$.
\end{proof}
\section{Code Generation and Examples}\label{sec:codeGen}
Using Isabelle's code generation mechanism~\cite{HaftmannCodeGeneration} we can
generate code for the \optimize{} function. Isabelle by default allows for the
generation of code in Haskell, SML, OCaml, and Scala. Using this generated code
we get the function \optimize{} in the objective language. Using Haskell, as an
example\footnote{Any of the aforementioned languages could be used.}, we can implement a simple parser to create a program that takes a
matrix and two vectors as input and calculates the solution to this linear
program using \optimize:\\

\begin{haskell}
  solveLP :: (String, String, String) -> Maybe ([Nat]+(Vec Rat))
  solveLP (a, b, c) = maximize matA vecB vecC where
    matrix = parseMatrix a
    mRows = (nat_of_int (maximum (map length matrix)))
    matA = mat_of_cols_list mRows matrix
    vecB = parseListToVec b
    vecC = parseListToVec c
\end{haskell}
\\

The compiled Haskell program can be used to solve \prettyref{ex:LP} with the
Constraints \ref{constraints:example} as input. The result is the vector
$[2\ 1]$. Hence $7*2+1*1 = 15$ is the maximum value. An example closer to our
intended application is \prettyref{ex:rps}.
\begin{example}[Solving Rock Paper Scissors]\label{ex:rps}
  The commonly known game of Rock-Paper-Scissors can be modelled with the
  following payoff matrix:
  \begin{center}
    \begin{tabularx}{0.5\textwidth}{ X | X | X | X |}

      &  Rock  & Paper & Scissors \\ \hline
      Rock     &  0     & -1     & 1       \\ \hline
      Paper    & 1     & 0     & -1       \\ \hline
      Scissors & -1    & 1    & 0        \\ \hline
    \end{tabularx}
  \end{center}
  This payoff matrix is to be interpreted as follows: If player one (rows) plays
  paper and player two (columns) plays rock then the payoff for player one has a
  payoff of $1$ meanwhile player two has the payoff of $-1$ (i.e. player one
  wins and player two loses). Since the sum of the payoffs always equals $0$,
  this is a zero sum game. Encoding the strategies rock, paper, and scissors, in
  $x_1$, $x_2$, and $x_3$ we can encode this game into the following linear
  program:
  \begin{align}
    \begin{split}
      \maximize&\ u \\
      \subjectto& u \leq - x_2 + x_3 \\
      & u \leq x_1 - x_3   \\
      & u \leq -x_1 + x_2 \\
      & x_1 + x_2 + x_3 = 1 \\
      & 0 \leq x_1,\ 0 \leq x_2,\ 0 \leq x_3
    \end{split}
  \end{align}
  The first three constraints encode the payoff matrix where $u$ is the payoff
  (i.e. utility). Hence, $u \leq - x_2 + x_3$ implies that the payoff $u$ cannot
  be higher than the sum of the utilities for playing rock combined
  ($0,\ -1,\ 1$). Similarly, we encode the other strategies. The last two lines
  of constraints ensure that the resulting assignment is a probability
  distribution over $\{x_1, x_2, x_3 \}$. After transforming these to matrix
  form we can use the extracted program to calculate the following vector:
  \[ [0, \frac{1}{3}, \frac{1}{3}, \frac{1}{3}]
  \]
  Meaning the expected payoff of playing the optimal strategy is $0$, and the
  optimal strategy is the mixed strategy of playing $x_1$, $x_2$, and $x_3$ with
  equal probability ($\frac{1}{3}$). Hence, playing rock, paper, or scissors
  each with a probability of $\frac{1}{3}$ is the optimal strategy of this game.
\end{example}

For small game theory examples such as \prettyref{ex:rps} the generated
algorithm performs well and produces a result instantly.  However, we did not
tamper with the underlying formalisation of the simplex
algorithm~\cite{Spasic:FormIncrSimplex} which famously has exponential worst
case complexity but behaves well most of the time.  Hence our algorithm exerts
the same asymptotic complexity as the underlying general simplex
algorithm. However, since we introduce different kinds of constraints the number
of constraints roughly doubles. Neither our reduction nor the underlying
algorithm was formalised with efficiency or competitiveness in constraint
solving in mind. Hence, we did not conduct any experiments comparing the
generated code with existing off the shelf constraint solvers. 
\section{Conclusion and Future Work}\label{sec:conclusion}
We presented the formalisation of an algorithm for solving linear programs. This
work is based on previous formalisation's of the general simplex algorithm as
well as a matrix library. The previous formalisation only considers the
satisfaction of linear constraints and does not allow for optimising an
objective function. We improved upon this by incorporating optimisation. Linear
programming (i.e. linear optimisation) has many applications in many different
fields. In our case the motivation is its use in game theory where linear
programming can be used to solve two player zero-sum games. In this paper we
present a formalisation of an algorithm that solves linear programs as well as
some results derived from it. Furthermore the algorithm is described in such a
way that Isabelle's code generation mechanism can be used to generate executable
code providing a verified solver for linear programs. The formalisation
is part of the Archive of Formal Proofs~\cite{Linear_Programming-AFP}.

Although the algorithm is formally proven to be sound within the proof
assistant, a completeness proof is sketched in this paper but does
not exist in a formalised manner, yet. We leave this for future work.
Furthermore, we plan on using this formalisation to produce a verified
solver for zero-sum two player games as part of a game
theory~\cite{jpck-itp18} formalisation effort.



\bibliography{biblio}


\end{document}